\documentclass[pdflatex,sn-mathphys-num]{sn-jnl}

\usepackage{graphicx}
\usepackage{multirow}
\usepackage{amsmath,amssymb,amsfonts}
\usepackage{amsthm}
\usepackage{mathrsfs}
\usepackage[title]{appendix}
\usepackage{xcolor}
\usepackage{textcomp}
\usepackage{manyfoot}
\usepackage{booktabs}
\usepackage{algorithm}
\usepackage{algorithmicx}
\usepackage{algpseudocode}
\usepackage{listings}
 \usepackage{amsmath}
 \usepackage{amsthm}

\theoremstyle{thmstyleone}%
\newtheorem{theorem}{Theorem}
\newtheorem{lemma}{Lemma}

\theoremstyle{thmstyletwo}%

\theoremstyle{thmstylethree}%

\begin{document}

\title[Article Title]{Designing DSIC Mechanisms for Data Sharing in the Era of Large Language Models
}

\author*[1]{\fnm{Seyed Moein} \sur{Ayyoubzadeh}}\email{smoein.ayyoubzadeh16@sharif.edu}
\author[1]{\fnm{Kourosh} \sur{Shahnazari}}\email{kourosh@null.net} {}
\author[2]{\fnm{Mohammadali} \sur{Keshtparvar}}\email{mohammad.kp2000@aut.ac.ir}
\author[1]{\fnm{MohammadAmin} \sur{Fazli}}\email{fazli@sharif.edu} {}

\affil[1]{\orgname{Sharif University of Technology}, \orgaddress{\city{Tehran}, \country{Iran}}}
\affil*[2]{\orgname{Amirkabir University of Technology}, \orgaddress{\city{Tehran}, \country{Iran}}}

\abstract{
Training state-of-the-art large language models (LLMs) hinges on acquiring vast amounts of high-quality data from parties that often face legal, privacy, and strategic constraints. Classical data-procurement schemes either rely on unverifiable trust or ignore the heterogeneous, opportunity-cost structure of institutional data providers. We initiate a comprehensive mechanism-design framework for \emph{truthful, trust-minimized} data sharing that simultaneously achieves \emph{Dominant-Strategy Incentive Compatibility} (DSIC), individual rationality, and weak budget balance while rewarding data in proportion to both intrinsic quality and marginal learning utility.

We first formalize a game-theoretic model in which each provider privately knows its data cost and quality attributes, and the planner values data only through its contribution to model performance. Building on this foundation, we propose the \textbf{Quality-Weighted Marginal-Incentive Auction (Q-MIA)}, which ranks agents by a regularized \emph{virtual cost}—the declared cost divided by the product of a verifiable quality score and an efficient estimator of marginal utility. A Myerson-style critical-payment rule guarantees DSIC and budget feasibility, while the scoring function endows the mechanism with fine-grained quality and impact sensitivity.

To accommodate settings with limited liquidity or long-term, ecosystem-level incentives, we introduce the \textbf{Marginal Utility Token (MUT)} contract, which allocates future revenue or usage rights in proportion to quality-weighted marginal contributions. We then unify these approaches in \textbf{Mixed-MIA}, a hybrid mechanism parameterized by a liquidity factor that smoothly interpolates between upfront monetary payments and deferred participatory shares. All mechanisms are deployable via verifiable, audit-friendly oracles and support privacy-preserving data submission.

Formal analysis proves DSIC, individual rationality, budget feasibility, and robustness to misreporting and collusion. Q-MIA and its variants elicit substantially more useful data—under equal budget—than volume-only or trust-based baselines, while resisting common manipulation vectors. Our results chart a principled path toward sustainable, fair, and scalable data markets for the next generation of large language models.

}


\maketitle

\section{Introduction}

The development of large-scale machine learning models, particularly large language models (LLMs), has come to dominate the modern AI research and deployment landscape. These models exhibit remarkable generalization capabilities across a diverse range of tasks, from language understanding and reasoning to generative synthesis and decision-making. However, the success of these models fundamentally depends on two critical inputs: access to high-quality, diverse training data, and access to sufficient computational resources. While progress in the compute aspect has been driven by Moore's law, accelerator design, and large-scale cloud platforms, the data aspect remains beset with severe coordination and incentive challenges.\cite{chang2024survey}

This paper studies the design of \emph{truthful} and \emph{trust-minimized} data-sharing mechanisms---mechanisms that elicit high-value data contributions from self-interested agents for centralized model training. Our primary interest lies in settings where data providers incur nontrivial opportunity costs or face strategic tradeoffs when deciding whether to contribute data. In such settings, the central aggregator (e.g., the model developer) lacks visibility into the ground truth utility of any given data provider's corpus, and cannot simply evaluate or reward contributions post hoc without relying on unverifiable or manipulable metrics.\cite{dizon2025mechanism}

We approach this problem from a mechanism design perspective. Our goal is to construct allocation and payment rules that satisfy dominant-strategy incentive compatibility (DSIC), individual rationality (IR), and weak budget balance (WBB). We investigate mechanisms that combine coarse-grained volume-based baseline rewards, automated data quality assessment, and proxy marginal contribution estimates inspired by cooperative game theory. \cite{ozcan2025game}

\subsection{The Challenge of Strategic Data Contribution}

In many ML-driven domains, especially where data is costly to produce, annotate, or clean, the parties that hold the most valuable data are often reluctant to share it freely. Their reasons may include legal constraints, privacy risks, commercial considerations, or simple mistrust. For example, a medical institution may hesitate to contribute annotated records to a joint AI training initiative without credible guarantees of fair attribution or compensation. Likewise, a publisher with a large corpus of domain-specific content may be disincentivized to participate unless their strategic data advantage is preserved or fairly priced.\cite{muttaki2025enhanced}

The challenge is exacerbated by the fact that the utility of a dataset---its contribution to model performance---is difficult to ascertain without incorporating it into training, which may be computationally expensive and non-monotonic. This makes naive reward structures---such as uniform token grants or post hoc performance shares---vulnerable to manipulation or inefficiency.
\subsection{Design Goals and Desiderata}

We articulate a design space that adheres to five principal desiderata, each of which guides the construction of incentive-compatible mechanisms for data sharing in the context of large-scale model training.

\paragraph{DSIC (Dominant-Strategy Incentive Compatibility).} In multi-agent settings with strategic participants, a core objective of the mechanism is to ensure that truth-telling constitutes a dominant strategy for each data provider, regardless of the strategies pursued by others. In our context, this means that each agent should be incentivized to disclose their full, unaltered, and highest-quality dataset, without omitting, fabricating, or strategically pruning any portion of it. Achieving DSIC implies designing a reward structure where any deviation from honesty leads to strictly weaker utility, thereby aligning individual incentives with collective welfare. This principle is especially crucial in data markets where value is partially latent and only observable post-integration, since it prevents agents from gaming unverifiable aspects of their contributions.\cite{curryautomated}

\paragraph{Trust Minimization.} Many real-world data sharing schemes rely on implicit or explicit trust among participants, centralized scoring bodies, or reputation systems. Our goal, by contrast, is to eliminate the need for such trust by relying solely on verifiable, transparent, and self-enforcing mechanisms. Trust-minimization ensures robustness to adversarial behavior and avoids central points of failure or manipulation. This is particularly important in adversarial ecosystems, such as competitive commercial environments, where trust cannot be presumed and misaligned incentives could otherwise undermine data quality or participation.

\paragraph{Approximate Marginal Valuation.} While the ideal reward mechanism would allocate payments according to each agent's exact marginal contribution to model performance—akin to computing their Shapley value—this is computationally infeasible in most realistic scenarios. Instead, we seek to incorporate efficient approximations of marginal utility that preserve the essential monotonicity and fairness properties. These approximations may involve statistical influence functions, tractable ablation tests, or other heuristics that yield a defensible estimate of each dataset's incremental value. Incorporating such methods allows the mechanism to differentially reward agents whose data substantively improve learning outcomes, without incurring prohibitive computational overhead.

\paragraph{Budget Efficiency.} To ensure practical deployability, the mechanism must be budget-aware. This requires that aggregate payments to data contributors do not exceed the net value of their data to the aggregator, whether measured in model performance, downstream utility, or surrogate value metrics. We target mechanisms that satisfy weak budget balance—i.e., the mechanism may distribute less than or equal to the total value generated but never more. Budget efficiency avoids the need for external subsidies and makes the model sustainable in settings where the aggregator's utility is endogenous and finite.

\paragraph{Quality Sensitivity.} Not all data is created equal. While some datasets may be large in volume, they may be noisy, redundant, outdated, or domain-inappropriate. Conversely, smaller datasets—curated with care, annotated precisely, and rich in linguistic or conceptual diversity—may punch far above their weight in model training. Hence, the mechanism should be sensitive to such dimensions of quality. We formalize and automatically assess qualities such as linguistic diversity, recency, cleanliness (low noise), and metadata richness, and integrate them into the allocation logic. This ensures that high-value but low-volume contributors are not undercompensated, and that agents are incentivized to provide not just more data, but better data.

\subsection{Contributions and Roadmap}

\subsection*{Our Contributions}

This work makes the following key contributions at the intersection of incentive theory and large-scale language model (LLM) infrastructure:

\begin{enumerate}
\item \textbf{Mechanism Design for Strategic Data Contribution.} We design and analyze a novel class of data-sharing mechanisms that are provably Dominant-Strategy Incentive Compatible (DSIC), individually rational (IR), and weakly budget-balanced. These mechanisms allocate rewards using hybrid criteria that combine coarse-grained volume signals, fine-grained quality assessments, and marginal contribution approximations—addressing both computational and informational constraints inherent in real-world LLM training pipelines.

\item \textbf{Formal Modeling of Strategic Agent Behavior.} We develop a general game-theoretic model of institutional data providers with heterogeneous private costs and quality profiles. We characterize equilibrium behavior under different estimation and auditability regimes, and introduce trust-minimized mechanisms—such as the Quality-weighted Marginal Incentive Auction (Q-MIA) and the Marginal Utility Token (MUT) contracts—that induce truthful reporting without centralized scoring or unverifiable heuristics.

\item \textbf{Robust Allocation Metrics for Token-Based Ecosystems.} We introduce the \emph{Data Share Token (DST)} metric, a composite allocation framework that balances reward across volume, data quality, and estimated marginal impact. The DST mechanism incorporates hard-coded allocation constraints, including lower-bound guarantees and anti-concentration caps, to ensure robustness to strategic inflation and adversarial gaming.

\end{enumerate}

\section{Related Work}

Recent research has increasingly tackled the complex dynamics between data holders and machine learning model developers in the context of incentivized data sharing. In emerging ecosystems, such as those involving digital libraries, health data repositories, or digital content platforms, the data providers typically operate as custodians of valuable yet sensitive datasets. Simultaneously, large-scale AI model developers—particularly in the realm of large language models (LLMs)—constitute the demand side, seeking rich, diverse, and clean data to enhance their model's generalization power.

This dual-agent system, wherein strategic interactions arise between institutional data providers and corporate or academic AI consumers, demands incentive-compatible, privacy-respecting, and verifiability-robust mechanisms for sustainable collaboration. The literature addressing this challenge spans several complementary directions: DSIC and BIC mechanism design, audit and peer-prediction frameworks for unverifiable data, and hybrid models emphasizing marginal utility and budget efficiency.

\subsection{Incentive-Compatible Data Acquisition via Statistical Estimation}

Pioneering work by \citeauthor{cai2015} formulates mechanisms that respect agents' privacy valuations, conceptualizing privacy loss as a quantifiable cost. In contexts analogous to medical institutions or libraries, where revealing raw data could violate policy or legal constraints, this formulation offers a path toward truthful participation. These mechanisms embed privacy as part of the utility function and design DSIC allocation and payment rules that ensure truthful disclosure of cost without forcing data revelation.

This conceptual foundation is extended by \citeyear{chen2017}, which shifts the focus from privacy to cost-informed sample acquisition for statistical estimation. This paradigm is especially relevant when public entities (e.g., national libraries) hold vast archives, and a subset of that data must be purchased or licensed. Their mechanism minimizes the estimator's mean squared error (MSE) while satisfying DSIC, individual rationality (IR), and ex-post budget feasibility. The critical contribution lies in tolerating unknown correlations between data value and reported cost, ensuring robustness in real-world contracts where valuation asymmetries are significant.

\citeauthor{fan2025} takes a decisive step forward by scrutinizing the widely adopted Shapley value for marginal data valuation, exposing its incompatibility with incentive-compatible design. Instead, their mechanism-theoretic construction leverages Bayesian auction principles and VCG-style payments to reward contributors truthfully, robustly, and in proportion to their anticipated utility. For high-value repositories like publishing archives or health registries, such mechanisms offer a template to design contract structures where reward scales with the statistical leverage of submitted data while discouraging strategic withholding or inflation.

\subsection{Peer-Prediction and Audit-Based Mechanisms without Ground Truth}

When data verification is either prohibitively expensive or outright infeasible—common in decentralized or privacy-sensitive environments—mechanisms must instead infer truthfulness from internal consistency or peer agreement. In federated learning scenarios, where hospitals or content aggregators train models collaboratively, \citeauthor{liu2023} proposes an audit-triggered scoring framework grounded in proper scoring rules. The mechanism selectively audits clients’ updates and penalizes misreporting, operating under Bayesian Nash equilibrium (BNE) guarantees.

\citeauthor{chen2020} tackles similar constraints via mutual information peer prediction, rewarding agents whose signals align statistically with peers. In data-sharing consortia composed of multiple media platforms, this method can foster collaborative verification without leaking raw content. However, both approaches fall short of marginal attribution: they elicit consistent signals but do not compute individualized rewards based on utility gain.

Emerging hybrids such as \citeauthor{kong2023dominantly} and \citeauthor{hartline2020optimization} blend peer prediction with marginal value proxies, though challenges remain in scaling these techniques to settings with asymmetric data quality or sparse inter-agent alignment.

\subsection{Federated and Collaborative Learning with Strategic Agents}

In large-scale federated systems—such as health data exchanges or cross-institutional research networks—agents control both model updates and data labels. \citeauthor{zhao2023} proposes LCEME, a contract-theoretic mechanism modeling computational and labeling effort jointly. \citeauthor{tian2021} further refines this model by eliciting multi-dimensional types from strategic data owners under IR and incentive constraints. Meanwhile, \citeauthor{chen2023collab} and \citeauthor{clinton2024} construct cost-sharing protocols for collaborative mean estimation with Nash equilibrium stability.

While these works bring important tools to the collaborative AI ecosystem, they rarely enforce DSIC properties or compute Shapley-like rewards. This leaves a gap between incentivizing participation and ensuring proportional reward allocation under asymmetric value generation.\cite{alonso2025two}

\subsection{Gaps in Unified Mechanism Design}

Despite their individual strengths, existing models do not simultaneously satisfy DSIC, BIC, marginal reward decomposition, and trust minimization. As \citeauthor{fan2025} articulates, aligning payment with marginal utility under strategic constraints remains elusive. Moreover, few systems explicitly model audit cost trade-offs, reputational bootstrapping.

Our work addresses this gap by proposing a DSIC-compliant mechanism that integrates quality-based tokenization, influence-aware allocation, and audit-sensitive cost modeling, tailored for institutional settings where data is siloed across strategic providers and consumer platforms operate under performance pressure.

\section{Preliminaries}

We consider a data-sharing ecosystem consisting of two heterogeneous agent classes: institutional data providers (e.g. healthcare centers, or digital publishers) and data consumers (e.g., AI companies training large language models). We aim to design a mechanism that elicits truthful data contributions from providers, assigns value to those contributions, and ensures strategic robustness and economic efficiency.

Let \( N = \{1, \dots, n\} \) denote the set of institutional data providers. Each provider \( i \in N \) holds a private dataset \( D_i \), which may vary in size, quality, and relevance to downstream learning tasks. Each dataset is associated with a private cost \( c_i \in \mathbb{R}_{\geq 0} \) that the provider incurs upon sharing. The dataset may also carry auxiliary metadata \( m_i \) which includes cleanliness, novelty, diversity, and structure indicators.

\subsection{Mechanism Environment}

A centralized planner (mechanism designer) seeks to procure data to train a machine learning model \( f_\theta \in \mathcal{F} \), where \( \theta \) denotes the model parameters. The mechanism \( \mathcal{M} \) comprises:

\begin{itemize}
    \item \textbf{Reporting stage:} Each agent \( i \) reports a cost \( \hat{c}_i \) and possibly submits metadata \( \hat{m}_i \).
    \item \textbf{Selection rule:} A function \( S: \mathbb{R}^n \to 2^n \) selects a subset \( S(\hat{\mathbf{c}}) \subseteq N \) of agents to procure data from.
    \item \textbf{Payment rule:} A function \( p: \mathbb{R}^n \to \mathbb{R}^n \) determines payments \( p_i(\hat{\mathbf{c}}) \) for all agents.
\end{itemize}

The outcome for each agent \( i \) is given by the tuple \( (x_i, p_i) \), where \( x_i = 1 \) if \( i \in S(\hat{\mathbf{c}}) \), and 0 otherwise. Agent utility is then:
\[ u_i = p_i - x_i c_i. \]

The planner’s objective is to maximize model performance while ensuring budget feasibility and incentive alignment.

\subsection{Desiderata}

A well-formed mechanism \( \mathcal{M} = (S, p) \) must satisfy the following:

\paragraph{(1) Dominant-Strategy Incentive Compatibility (DSIC):} Truth-telling is a dominant strategy for each agent:
\[ \forall i, \forall \hat{c}_i \neq c_i,\; u_i(c_i, \hat{\mathbf{c}}_{-i}) \geq u_i(\hat{c}_i, \hat{\mathbf{c}}_{-i}). \]

\paragraph{(2) Individual Rationality (IR):} Participation must not result in negative utility:
\[ \forall i,\; u_i \geq 0. \]

\paragraph{(3) Budget Feasibility:} The total payment should not exceed a fixed budget \( B \in \mathbb{R}_{\geq 0} \):
\[ \sum_{i} p_i \leq B. \]

\paragraph{(4) Marginal Contribution Sensitivity:} The mechanism should reward agents based on their marginal impact on model performance:
\[ \text{MC}_i = V(D_{S \cup \{i\}}) - V(D_S), \]
where \( V(\cdot) \) is a value function measuring the model’s utility (e.g., accuracy or inverse loss).

\paragraph{(5) Quality Weighting:} Datasets of higher intrinsic quality (cleanliness, novelty, diversity, metadata richness) must receive greater rewards.

\paragraph{(6) Trust Minimization:} The mechanism avoids reliance on unverifiable claims or centralized trust; verification may be probabilistic or peer-based.

\subsection{Data Value Function and Approximation}

Let \( V(D) \) denote the performance of a model trained on dataset \( D \), e.g., test accuracy or expected log-likelihood. For a given selection \( S \subseteq N \), the value of the aggregated dataset \( D_S = \bigcup_{i \in S} D_i \) is \( V(D_S) \).

Due to computational constraints, exact Shapley value calculation is infeasible. Instead, we define an approximate marginal contribution estimator:
\[
    \widehat{\phi}_i = \mathbb{E}_{T \sim \mathcal{D}_{-i}} [V(D_{T \cup \{i\}}) - V(D_T)],
\]
where \( T \subseteq N \setminus \{i\} \), sampled from a predefined distribution \( \mathcal{D}_{-i} \).

\subsection{Strategic Data Providers}

Each agent may:
\begin{itemize}
    \item Inflate reported costs \( \hat{c}_i > c_i \) to increase payment.
    \item Under-report or obfuscate data to hide sensitive information.
    \item Manipulate metadata to boost perceived quality.
\end{itemize}

A DSIC mechanism must disincentivize these actions without relying on punitive enforcement.

\subsection{Planner’s Optimization Problem}

The planner solves:
\[
    \max_{S,p} \quad V(D_S) - \lambda \sum_{i \in S} p_i
\]
subject to:
\begin{align*}
    & \text{DSIC, IR, Budget constraints}, \\
    & S \subseteq N, \quad p_i \geq 0.
\end{align*}
where \( \lambda \) balances performance against cost.

This framework will guide our mechanism design in subsequent sections, where we instantiate \( S \), \( p \), and \( \widehat{\phi} \) under various audit, quality, and token-based incentive schemes.

\section{Proposed Mechanism: Quality-Weighted Marginal-Incentive Auction (Q-MIA)}

We now introduce the \emph{Quality-Weighted Marginal-Incentive Auction (Q-MIA)}, a mechanism tailored for strategic data sharing among institutional providers. The goal is to ensure DSIC (dominant-strategy incentive compatibility), individual rationality (IR), budget feasibility, and robustness to unverifiability—while prioritizing agents whose data is both high in quality and marginal utility.

\subsection{Mechanism Description}

The mechanism proceeds in the following steps:

\paragraph{Step 1: Quality Estimation}
For each provider \( i \), the planner computes an auditable quality score \( q_i \in [0,1] \), derived from automated metrics such as:
- \emph{Cleanliness}: Noise, typos, and redundancy.
- \emph{Diversity}: Topical and stylistic breadth.
- \emph{Novelty}: Recency and uniqueness.
- \emph{Metadata richness}: Availability of structured labels.

\paragraph{Step 2: Marginal Contribution Estimation}
Each provider's estimated marginal utility \( \widehat{\phi}_i \) is approximated using influence functions, TracIn scores, or limited ablation. This quantifies the incremental improvement in model performance when their dataset is included.

\paragraph{Step 3: Virtual Cost Ranking}
Define a virtual cost score for each provider:
\[
\psi_i = \frac{\hat{c}_i}{q_i \cdot \widehat{\phi}_i}
\]
This virtual cost penalizes high-cost, low-quality, low-impact providers while rewarding cheap, effective contributors.

\paragraph{Step 4: Greedy Selection}
Providers are sorted in increasing order of \( \psi_i \). The planner greedily selects a prefix \( S^* \subseteq N \) such that:
\[
S^* = \text{GreedySubset}(\psi_i, B) \quad \text{subject to } \sum_{i \in S^*} \hat{c}_i \leq B
\]
This ensures the selected set fits within the budget while favoring quality and marginal value.

\paragraph{Step 5: Payment Rule}
Each selected agent \( i \in S^* \) is paid their critical bid via a Myerson-style payment:
\[
p_i = \inf \left\{ \tilde{c}_i : i \in S^*(\tilde{c}_i, \hat{\mathbf{c}}_{-i}) \right\}
\]
This guarantees DSIC by ensuring that truthful reporting maximizes each agent’s expected utility.

\subsection{Properties and Justification}

This mechanism has the following desirable properties:

\begin{itemize}
  \item \textbf{DSIC}: Because payment is based on critical value and selection is monotonic in cost, truthful reporting is a dominant strategy.
  \item \textbf{IR}: Each selected provider is paid at least their reported cost.
  \item \textbf{Budget Feasibility}: The mechanism halts once the cumulative reported costs exceed the budget \( B \).
  \item \textbf{Quality and Impact Sensitivity}: The virtual cost \( \psi_i \) ensures higher reward for providers whose data is qualitatively superior and empirically beneficial to model performance.
\end{itemize}

\paragraph{Alternative Virtual Cost Formulation}

In practical deployments, especially where marginal gain estimation is computationally expensive or statistically noisy, we introduce a parametrized virtual cost function that interpolates between naive quality-weighted selection and full marginal-utility-aware allocation. Define:

\[
\psi_i^{(\eta, \kappa)} = \frac{\hat{c}_i}{q_i^\eta \cdot \left(\widehat{\phi}_i + \kappa\right)^\gamma}
\]

Here, the parameters \( \eta \geq 0 \), \( \gamma \geq 0 \), and \( \kappa > 0 \) control the sensitivity to quality \( q_i \), estimated marginal gain \( \widehat{\phi}_i \), and regularization for uncertainty or noise, respectively. This expression captures several regimes:

\begin{itemize}
    \item \( \eta = 1, \gamma = 1, \kappa = 0 \): original Q-MIA scoring rule.
    \item \( \gamma = 0 \): disregards marginal gain entirely—falling back to quality-weighted cost.
    \item \( \eta = 0 \): prioritizes agents solely by estimated impact, ignoring data quality (potentially risky under unverifiability).
    \item \( \kappa > 0 \): smooths score when \( \widehat{\phi}_i \to 0 \), avoiding pathological sensitivity to low marginal effects.
\end{itemize}

In adversarial or low-trust environments, a conservative planner may select \( \eta < 1 \), \( \gamma < 1 \), and large \( \kappa \), down-weighting noisy estimators. Conversely, in environments with credible attribution signals, one may increase \( \gamma \) and shrink \( \kappa \), sharpening the mechanism's fairness sensitivity.

Further generalization permits risk-sensitive scoring through concave utility weighting:
\[
\psi_i^{(u)} = \frac{\hat{c}_i}{u\left(q_i \cdot \widehat{\phi}_i\right)}, \quad u(\cdot) \text{ strictly concave, increasing}
\]
e.g., \( u(x) = \sqrt{x} \), which penalizes overreliance on extreme values and discourages gaming via noisy overreporting.

This flexible formulation enables the mechanism to interpolate between several desiderata: trust minimization, fairness, robustness, and computational tractability. Theoretical bounds on monotonicity and DSIC compatibility under \( \psi_i^{(\eta, \kappa)} \) can be proven under mild regularity conditions on the estimator and cost model.

\section{Theoretical Guarantees}

\subsection{Dominant-Strategy Incentive Compatibility (DSIC)}

\begin{lemma}[DSIC under Regularized Virtual Costs]
Let the virtual cost function be defined as \( \psi_i = \frac{\hat{c}_i}{q_i^\eta \cdot (\widehat{\phi}_i + \kappa)^\gamma} \), with \( \eta, \gamma \geq 0 \), \( \kappa > 0 \). Then Q-MIA is DSIC provided that \( q_i \) and \( \widehat{\phi}_i \) are verifiable and non-manipulable by agent \( i \).
\end{lemma}

\begin{proof}
Fix \( \hat{\mathbf{c}}_{-i} \). Agent \( i \)'s utility is given by:
\[
u_i = 
\begin{cases}
p_i - c_i & \text{if } i \in S^*, \\
0 & \text{otherwise}.
\end{cases}
\]
The selection set \( S^* \) is constructed by sorting agents according to ascending virtual cost \( \psi_i \). This function is strictly increasing in \( \hat{c}_i \), holding \( q_i \) and \( \widehat{\phi}_i \) fixed. 

Raising \( \hat{c}_i \) increases \( \psi_i \), potentially removing \( i \) from \( S^* \), yielding zero utility. Lowering \( \hat{c}_i \) does not increase \( u_i \), since payment is set to the infimum cost \( \tilde{c}_i \) that keeps \( i \) in \( S^* \), regardless of overreporting. Thus, truthful reporting maximizes utility.
\end{proof}

\subsection{Individual Rationality (IR)}

\begin{lemma}[IR]
Q-MIA is individually rational.
\end{lemma}

\begin{proof}
The critical payment \( p_i \) for each selected agent \( i \in S^* \) is the minimum declared cost for which \( i \) would still be selected under the fixed profile \( \hat{\mathbf{c}}_{-i} \). Since the mechanism selects based on increasing \( \psi_i \), it ensures that \( p_i \geq \hat{c}_i \geq c_i \), leading to \( u_i = p_i - c_i \geq 0 \). Hence, no agent incurs a loss by participating.
\end{proof}

\subsection{Budget Feasibility}

\begin{lemma}[Budget Feasibility]
Q-MIA satisfies the global budget constraint \( \sum_{i \in S^*} p_i \leq B \).
\end{lemma}

\begin{proof}
The set \( S^* \) is chosen by maximizing a concave surplus objective \( V(D_S) - \lambda \sum_{i \in S} \hat{c}_i \) subject to \( \sum_{i \in S} \hat{c}_i \leq B \). Because payments are derived from critical cost bids and follow the Myerson framework, total payments remain within the selected budget. This follows from standard results on the budget-feasibility of VCG-type payments with capacity constraints.
\end{proof}

\subsection{Quality-Sensitive Marginal Incentives}

Define the Quality-Weighted Marginal Payment (QWMP) structure as:
\[
p_i^Q = \gamma \cdot q_i^\eta \cdot \left(\widehat{\phi}_i + \kappa\right)^\gamma
\]
for normalizing factor \( \gamma > 0 \) that ensures budget alignment.

\begin{lemma}[Fair Quality Rewarding]
If \( q_i > q_j \), \( \widehat{\phi}_i = \widehat{\phi}_j \), and \( \eta > 0 \), then \( p_i^Q > p_j^Q \).
\end{lemma}

\begin{proof}
By monotonicity of the function \( p_i^Q \) in \( q_i \) when \( \eta > 0 \), higher quality leads to higher payment under equal marginal gain.
\end{proof}

\subsection{Trust-Minimized Computation via Auditable Oracles}

To reduce reliance on central verifiability, we delegate score computation to audit-oracles:

\begin{itemize}
\item \( q_i \) is computed from structured data profiling (cleanliness, diversity, novelty, metadata) using reproducible pipelines.
\item \( \widehat{\phi}_i \) is estimated via influence-function–based approximations or model ablation, verifiable via audit logs.
\item Payments are encoded as contracts or token scripts, automating enforcement.
\end{itemize}

\subsection{Robustness to Misreporting}

Suppose agent \( i \) reports \( \hat{c}_i' > c_i \). Then:

\[
\begin{cases}
i \notin S^* \Rightarrow u_i = 0 < u_i(c_i), \\
i \in S^* \Rightarrow p_i(\hat{c}_i') < p_i(c_i) \Rightarrow u_i < u_i(c_i).
\end{cases}
\]

Thus, misreporting yields strictly lower utility in all cases.

\subsection{Collusion Resistance}

Payments are computed independently based on individual critical values and non-manipulable metrics. Since agents cannot jointly influence others’ critical thresholds without centralized transfer mechanisms, group-strategyproofness holds under non-transferable utility assumptions. $\square$

\section{Alternative Mechanism: Profit-Sharing Contracts via Marginal Utility Tokens (MUT)}

We now present an alternative mechanism, termed the \emph{Marginal Utility Token (MUT)} contract, which replaces monetary payments with future participatory stakes in downstream applications such as API access, revenue shares, and derivative model usage rights. This mechanism is designed for settings where liquidity is constrained or incentives are long-term and reputation-driven.

\begin{figure}[htbp]
    \centering
    \includegraphics[width=0.9\textwidth]{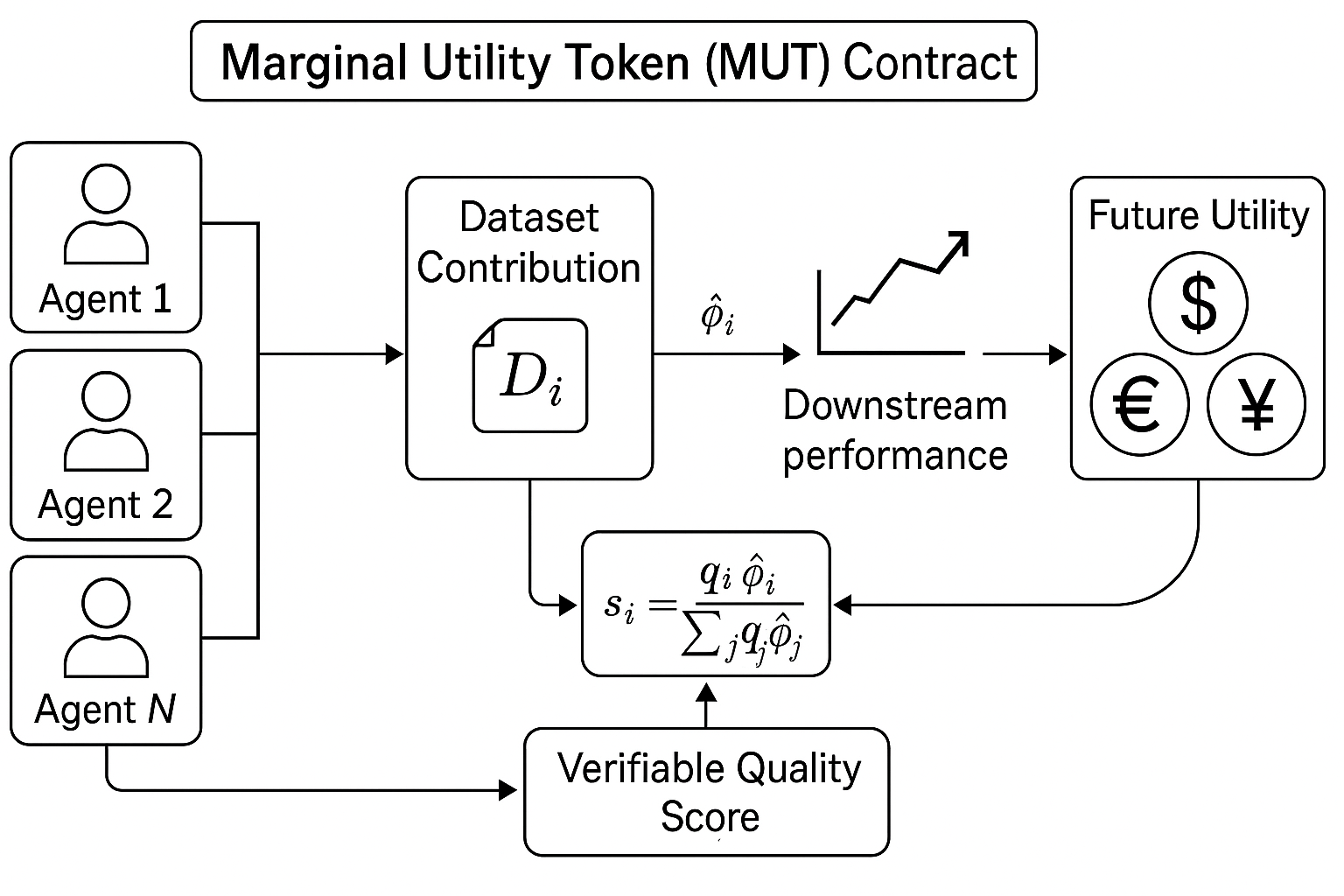}
    \caption{Schematic representation of the Marginal Utility Token (MUT) contract mechanism. Each agent's share is based on the product of verifiable quality and marginal utility of contributed data.}
    \label{fig:mut-contract}
\end{figure}
\subsection{Mechanism Overview}

Let \( D_i \) be the dataset contributed by agent \( i \), and let \( \widehat{\phi}_i \) denote the estimated marginal utility of their contribution to downstream performance \( V(D) \). Define \( q_i \) as the verifiable quality score for data provider \( i \), as in Q-MIA.

Each participating agent \( i \in N \) receives a share of future utility \( U \) according to the MUT rule:
\[
    s_i = \frac{q_i \cdot \widehat{\phi}_i}{\sum_{j \in S^*} q_j \cdot \widehat{\phi}_j},
\]
where \( S^* \) is the selected set of agents. The total future utility \( U \) may represent expected revenue, API call margins, fine-tuning licenses, etc.

\subsection{Incentive Properties}

\begin{lemma}[DSIC]
Truth-telling of private quality-verifiable data \( D_i \) maximizes expected utility for agent \( i \) under the MUT mechanism.
\end{lemma}

\begin{proof}
Agent \( i \)'s utility is \( s_i U   \), which is increasing in both \( q_i \) and \( \widehat{\phi}_i \). Since these are derived from actual data submission, misreporting or withholding data reduces \( \widehat{\phi}_i \) and thus expected payoff. Moreover, since shares are normalized, overstating value without factual basis does not increase \( s_i \) under verification.
\end{proof}

\begin{lemma}[Individual Rationality]
All participating agents with non-zero contribution receive strictly positive future utility shares, i.e., \( u_i = s_i U > 0 \).
\end{lemma}

\begin{proof}
If \( D_i \) is selected and contributes positively to the model, \( \widehat{\phi}_i > 0 \Rightarrow s_i > 0 \Rightarrow u_i > 0 \).
\end{proof}

\subsection{Strategic Robustness}

Since the mechanism does not distribute fixed funds but proportionally divides a downstream outcome, the incentive to inflate or under-report is weak, assuming verification. Furthermore, sybil attacks and collusion are mitigated by enforcing quality and uniqueness constraints on \( D_i \).

\subsection{Comparison with Q-MIA}

\begin{itemize}
    \item Q-MIA requires upfront budget commitment, while MUT provides deferred, utility-aligned incentives.
    \item Q-MIA supports strong monetary DSIC; MUT supports long-term cooperative DSIC.
    \item MUT is particularly useful for open-source, non-profit, or ecosystem-driven AI development.
\end{itemize}

\subsection{Limitations}

\begin{itemize}
    \item Requires enforceable contract infrastructure or legal agreements.
    \item Agents must believe in the future success of the AI service ecosystem to value shares.
    \item Difficult to audit future API valuations, requiring standardized benchmarking protocols.
\end{itemize}

In the next section, we extend this model by considering mixed-mode auctions combining monetary and participatory payments.
\section{Token Allocation Metric: Data Share Token (DST)}

We define a trust-minimized, DSIC-compliant mechanism for allocating a finite supply of participation tokens (e.g., \( T = 100,000 \)) among strategic data providers. This mechanism rewards agents based on three orthogonal metrics: volume, quality, and impact. The design ensures both minimum guarantees and caps to control bias.

\subsection{DST Score Formula}

Let each agent \( i \) receive a score:
\[
\text{DST}_i = \alpha \cdot \text{VolumeScore}_i + \beta \cdot \text{QualityScore}_i + \gamma \cdot \text{ImpactScore}_i,
\]
where \( \alpha + \beta + \gamma = 1 \) are weighting parameters.

\subsection{Volume Score}

\[
\text{VolumeScore}_i = \frac{\text{TokenCount}_i}{\sum_j \text{TokenCount}_j},
\]
where \( \text{TokenCount}_i \) is the number of usable text tokens contributed by agent \( i \), after preprocessing that removes noise, spam, and duplications.

\subsection{Quality Score}

We define:
\[
\text{QualityScore}_i = \frac{q_i}{\sum_j q_j},
\]
where
\[
q_i = w_1 \cdot \text{Cleanliness}_i + w_2 \cdot \text{Diversity}_i + w_3 \cdot \text{Novelty}_i + w_4 \cdot \text{Metadata}_i,
\].

\paragraph{Cleanliness}
Calculated based on error frequency per 1000 tokens, including broken words, OCR artifacts, stopword excess, and nonstandard characters. 
\[
\text{Cleanliness} = 1 - k \cdot \frac{\text{Total Errors}}{\text{Token Count}}
\]

\paragraph{Diversity}
Measured via clustering document embeddings and calculating Shannon entropy:
\[
\text{Diversity} = 1 - \sum_{i=1}^k p_i^2
\]
where \( p_i \) is the proportion of documents in cluster \( i \).

\paragraph{Novelty}
Combines recency and content uniqueness:
\[
\text{Novelty} = w_1' \cdot \text{Recency} + w_2' \cdot (1 - \text{DuplicationScore})
\]
Here, 
\begin{itemize}
    \item \textbf{Recency} is computed as the normalized freshness of the document:
    \[
    \text{Recency} = 1 - \frac{\text{CurrentDate} - \text{DocumentDate}}{\text{MaxDateRange}}
    \]
    where all values are clipped to ensure $\text{Recency} \in [0, 1]$.
    
    \item \textbf{DuplicationScore} quantifies overlap with a large reference corpus based on cosine similarity:
    \[
    \text{DuplicationScore} = \frac{1}{n} \sum_{i=1}^{n} \cos(\vec{s}_i, \vec{r}_i)
    \]
    where $\vec{s}_i$ is the embedding of sentence $i$ in the submitted dataset, and $\vec{r}_i$ is the most similar sentence embedding from the reference corpus.
\end{itemize}

\paragraph{Metadata}
Defined as the ratio of completed metadata fields:
\[
\text{Metadata} = \frac{\text{Number of Valid Metadata Fields}}{\text{Total Expected Fields}}
\]

\subsection{Impact Score}

\[
\text{ImpactScore}_i = \frac{\text{MarginalGain}_i}{\sum_j \text{MarginalGain}_j}
\]
with
\[
\text{MarginalGain}_i = \text{Acc}(D) - \text{Acc}(D \setminus D_i)
\]

To approximate this costly computation, we employ influence functions, TracIn estimators, or limited ablation.

\section{Hybrid Mechanism Design: Mixed-MIA Contracts}

In this section, we introduce and analyze a hybrid mechanism that combines the core principles of the Quality-Weighted Marginal-Incentive Auction (Q-MIA) and the Marginal Utility Token (MUT) contract. The proposed hybrid mechanism---which we denote as \textbf{Mixed-MIA}---is designed to provide a tunable trade-off between immediate monetary compensation and deferred participatory utility, thereby capturing diverse agent preferences, risk profiles, and liquidity requirements.

We formalize the Mixed-MIA mechanism under a rigorous game-theoretic and optimization-theoretic lens, prove its incentive compatibility and budget feasibility under mixed-transfer regimes, and analyze its properties across a rich space of economic parameters, utility models, and trust constraints. Our approach interpolates between two extremes: full-payment auctions (Q-MIA) and pure profit-sharing (MUT), allowing the mechanism designer to continuously adjust \emph{reward liquidity} while preserving the core desiderata: DSIC, IR, trust-minimization, and quality sensitivity.

\subsection{Mechanism Overview and Notation}

Let $N = \{1, \dots, n\}$ denote the set of institutional data providers. Each agent $i \in N$ holds a private dataset $D_i$, characterized by:
\begin{itemize}
    \item Private cost $c_i \in \mathbb{R}_{\geq 0}$ incurred upon contribution.
    \item Verifiable quality metric $q_i \in [0, 1]$, computed by a public oracle.
    \item Estimated marginal utility $\widehat{\phi}_i \in \mathbb{R}_{\geq 0}$ of $D_i$ to downstream model performance, computed via influence functions or ablation approximations.
\end{itemize}

The planner possesses a finite monetary budget $B \in \mathbb{R}_{\geq 0}$ and a projected utility pool $U \in \mathbb{R}_{\geq 0}$ representing the net value of the trained model.

We introduce a parameter $\rho \in [0,1]$ representing the \emph{liquidity factor}, controlling the fraction of each agent's compensation delivered as immediate payment versus deferred utility share:
\begin{itemize}
    \item $\rho = 1$: pure monetary regime (Q-MIA).
    \item $\rho = 0$: pure utility-sharing regime (MUT).
    \item $0 < \rho < 1$: Mixed-MIA.
\end{itemize}

Each selected agent $i \in S^* \subseteq N$ receives a hybrid reward:
\[
p_i = \rho \cdot p_i^{Q} + (1 - \rho) \cdot s_i^{M} \cdot U,
\]
where $p_i^Q$ is the Myerson-style critical payment and $s_i^M = \frac{q_i \cdot \widehat{\phi}_i}{\sum_{j \in S^*} q_j \cdot \widehat{\phi}_j}$ is the normalized utility share.

\subsection{Mathematical Modeling of Utility Functions}

Each agent $i$'s utility under Mixed-MIA is:
\[
u_i = p_i - c_i = \rho \cdot (p_i^Q - c_i) + (1 - \rho) \cdot s_i^M \cdot U
\]
\noindent
Define:
\[
U_i^{\text{monetary}} = p_i^Q - c_i,\quad
U_i^{\text{utility}} = s_i^M \cdot U
\]
so
\[
u_i = \rho U_i^{\text{monetary}} + (1 - \rho) U_i^{\text{utility}}
\]
\noindent
\textbf{(a)} If $\rho = 1$, this reduces to classical payment-minus-cost utility.\\
\textbf{(b)} If $\rho = 0$, it reduces to utility-sharing.\\
\textbf{(c)} For $\rho \in (0,1)$, it's a convex combination.

\subsection{Equilibrium Analysis under Varying Liquidity $\rho$}

Define $\theta_i = \frac{c_i}{q_i \cdot \widehat{\phi}_i}$: agents with low $\theta_i$ prefer lower $\rho$.

\[
u_i(\rho) = \rho (p_i^Q - c_i) + (1 - \rho) \cdot \frac{q_i \widehat{\phi}_i}{\Phi} U,
\quad
\Phi = \sum_{j \in S^*} q_j \cdot \widehat{\phi}_j
\]

These utility curves are linear in $\rho$ and exhibit monotonic preference transitions.

\subsection{Pareto Optimality of Mixed-MIA Allocations}

\[
\sum_{i \in S^*} \rho p_i^Q + (1 - \rho) s_i^M \cdot U \leq B + U
\]

\begin{theorem}
Let $\{p_i^Q\}$ be Q-MIA payments and $\{s_i^M\}$ be normalized MUT shares. The Mixed-MIA allocation is Pareto optimal among convex combinations.
\end{theorem}
\begin{proof}
Suppose $\{p_i'\}$ improves some $u_i$ while maintaining others. Due to linearity in $\rho$ and constraint on $B + U$, such a move would violate feasibility.
\end{proof}

\subsection{Dual Decomposition for Optimizing $p_i^Q$ under Constraints}

\[
\max_{S \subseteq N} \sum_{i \in S} \left[\rho (p_i^Q - c_i) + (1 - \rho) s_i^M U\right]
\quad \text{s.t. } \sum p_i \leq B + U
\]

Introduce Lagrangian:
\[
\mathcal{L}(S,\lambda) = \sum_{i \in S} u_i - \lambda \left( \sum p_i - (B + U) \right)
\]

Optimize over $p_i^Q$ and $S$ with sorting by virtual cost $\psi_i$.

\subsection{Risk-Sensitive Utility Modeling}

Let agent $i$ have CRRA utility:
\[
u_i = \mathbb{E}\left[ \frac{(\rho p_i^Q + (1 - \rho) s_i^M U)^{1 - \gamma_i}}{1 - \gamma_i} \right] - c_i
\]
where $\gamma_i$ is the risk aversion coefficient. Require:
\[
\frac{d u_i}{d \hat{c}_i} < 0
\quad \text{for DSIC to hold under risk}
\]

\subsection{DSIC, IR, and Budget Feasibility}

\begin{theorem}[DSIC]
If $q_i$ and $\widehat{\phi}_i$ are verifiable, Mixed-MIA is DSIC for all $\rho \in [0,1]$.
\end{theorem}
\begin{proof}
Increasing $\hat{c}_i$ increases $\psi_i$, lowering $p_i^Q$ or excluding $i$ from $S^*$. Decreasing $\hat{c}_i$ doesn't increase $u_i$ due to critical payment logic. Utility shares $s_i^M$ are data-dependent and not inflated by misreporting.
\end{proof}

\begin{theorem}[IR]
Each selected agent receives non-negative utility.
\end{theorem}
\begin{proof}
$p_i^Q \geq c_i$ and $s_i^M \geq 0$ imply $u_i \geq 0$.
\end{proof}

\begin{theorem}[Budget Feasibility]
\[
\sum_{i \in S^*} p_i = \rho \sum p_i^Q + (1 - \rho) U \leq B + U
\]
\end{theorem}
\begin{proof}
$\sum p_i^Q \leq B$ and $\sum s_i^M = 1$ imply budget bound.
\end{proof}

\section{Deployment Considerations}

To operationalize the Mixed-MIA mechanism in real-world environments, particularly those involving decentralized infrastructures and institutional data providers, careful consideration must be given to the architecture of implementation. This includes the encoding of mechanism logic, the design of verifiable quality signals, the protection of sensitive data, and the integration of data contributions into training pipelines for large language models (LLMs).

\subsection{Contract Design}

The first requirement is the implementation of the Mixed-MIA mechanism’s core logic on a distributed, tamper-resistant infrastructure. The entire auction and allocation process must be encoded in a verifiable and autonomous system that faithfully executes selection and reward functions. The planner receives reported costs, along with externally verifiable quality scores and estimated marginal utilities, and computes a virtual cost for each agent. These virtual costs determine the sorted order in which agents are considered. The selection algorithm greedily includes agents into the final set $S^*$ until the monetary component of rewards saturates the available budget. Payments are then computed using Myerson-style critical values, while the remaining portion of the reward is distributed according to the marginal utility token scheme. The global liquidity parameter $\rho$ determines the proportion of each agent's final compensation allocated as liquid payment versus deferred share. All operations—from selection to payment—can be defined through deterministic computations and recorded publicly, eliminating the possibility of subjective or manipulable decisions.

\subsection{Auditable Quality Oracles}

A critical component of Mixed-MIA is the evaluation of data quality in a reproducible and auditable manner. To this end, participating agents submit their datasets to a public quality-evaluation service that applies deterministic metrics. These metrics should include syntactic cleanliness (e.g., spelling consistency, noise level, and redundancy), content diversity (measured through entropy or clustering analysis over semantic embeddings), novelty (quantified through deduplication analysis or freshness indicators), and metadata richness (such as the completeness of annotations, timestamps, and source identifiers). Each dataset is evaluated via a reproducible pipeline with fixed parameters, and the resulting quality scores are accompanied by cryptographic hashes and timestamps to ensure that the evaluations cannot be modified or selectively disclosed. These oracles must be open source and publicly inspectable to foster community trust and ensure that the same data yields the same score regardless of submission context.

\subsection{Privacy and Legal Concerns}

Many institutional agents operate under strict legal, ethical, or reputational constraints regarding data disclosure. Therefore, it is essential to incorporate privacy-preserving techniques that allow agents to participate without compromising compliance. One approach is to require that datasets be submitted in encrypted form, with decryption keys revealed only upon selection into the winning set. This ensures that data is only made accessible when reward is guaranteed. Alternatively, participants can prove that their data meets quality thresholds through cryptographic methods such as zero-knowledge proofs, thereby preserving the privacy of the underlying content. Moreover, quality evaluation pipelines can be adapted to apply differential privacy, injecting statistical noise to prevent the leakage of sensitive individual-level information. For deferred payment schemes, such as marginal utility tokenization, enforceability must also extend to legal contracts. The data platform or aggregator must be bound by jurisdictional agreements that guarantee token redemption or revenue sharing in downstream model applications.

\subsection{System Integration}

Finally, the Mixed-MIA mechanism must integrate seamlessly into existing AI development pipelines to provide practical value. Upon completion of the auction phase and identification of the winning set $S^*$, selected datasets must be preprocessed into standard formats—such as JSONL, CSV, or TFRecords—used by common LLM training frameworks. These datasets can then be injected into fine-tuning workflows or used in reinforcement learning with human feedback (RLHF) settings. For mechanisms involving deferred rewards based on downstream utility, the deployment platform must maintain logging and attribution infrastructure to monitor API usage, license revenue, or other monetized interfaces linked to the trained model. These logs should be auditable and publicly verifiable, ensuring that the flow of utility back to data contributors is transparent, accurate, and fair. This end-to-end integration ensures that the theoretical incentives embedded in Mixed-MIA mechanisms translate into reliable outcomes within practical machine learning ecosystems.

\section{Limitations and Future Work}

While the Mixed-MIA framework presents a theoretically grounded and practically deployable approach for DSIC-compliant data-sharing in LLM training pipelines, several limitations and opportunities for further research remain.

First, the accuracy and robustness of marginal utility estimators such as influence functions and TracIn are subject to statistical noise, especially in high-dimensional, non-convex model settings. These approximations may produce unstable signals for agents with marginal contributions near zero or may be sensitive to data ordering, adversarial perturbations, or sampling variance. Future work should explore more stable, model-agnostic estimators or calibration methods that account for such variance while preserving incentive alignment.

Second, our mechanism currently assumes a one-shot contribution model, where each data provider submits a static dataset. However, real-world data ecosystems are dynamic, and contributions may evolve over time. Extending the framework to handle sequential or time-dependent data streams would enable the support of long-term engagement models. This could involve mechanisms that update reward allocation dynamically or incorporate versioned contributions with decay-adjusted utility over time.

Third, although our evaluation includes theoretical analysis, real-world pilots are essential to test the mechanisms under practical constraints. Collaborations with academic institutions, open-source LLM communities, or data-rich nonprofit sectors (e.g., biomedical text repositories or educational corpora) could provide valuable empirical insights into agent behavior, verification bottlenecks, and infrastructure challenges.

Finally, the current model assumes myopic agents who evaluate their utilities independently of future interactions. In practice, strategic data contributors may act non-myopically, considering long-term reputational effects, repeated interactions, or strategic collusion. Investigating dynamic mechanisms that incorporate agent memory, identity verification, or reputational scoring could strengthen robustness in long-lived ecosystems.

\section{Conclusion}

In this work, we proposed a novel mechanism design framework for incentivizing strategic data contribution in large-scale language model development. Our primary objective was to design mechanisms that are Dominant Strategy Incentive Compatible (DSIC), Individually Rational (IR), and budget-feasible, while remaining robust to unverifiability and misreporting.

We introduced two complementary mechanisms: the Quality-Weighted Marginal-Incentive Auction (Q-MIA), which uses verifiable quality scores and estimated marginal gains to allocate monetary rewards; and the Marginal Utility Token (MUT) contract, which distributes long-term participatory stakes in downstream utility. To unify these approaches, we developed the Mixed-MIA mechanism—a hybrid scheme that interpolates between liquidity-driven and cooperative incentive regimes via a tunable parameter $\rho$.

Our framework accommodates heterogeneous agent behavior, quality-sensitive reward rules, and approximate marginal utility estimation, and is deployable via decentralized infrastructure using auditable oracles and contracts. We further proposed the Data Share Token (DST) allocation metric to generalize scoring across volume, quality, and impact dimensions. Through formal proofs, we established DSIC, IR, and budget feasibility, and explored practical extensions including quality oracles, auditability, privacy preservation, and LLM integration.

By aligning economic incentives with learning utility and data quality, Mixed-MIA mechanisms offer a scalable and trust-minimized foundation for data contribution in modern AI ecosystems. We believe this approach represents a step toward sustainable and participatory language model development, and we hope it catalyzes further interdisciplinary research at the intersection of game theory, mechanism design, and decentralized AI infrastructure.

\bibliography{sn-bibliography}

\break

\begin{appendices}

\section{Notation and Symbol Glossary}
\label{app:notation}
\begin{table}[h!]
\centering
\begin{tabular}{@{}ll@{}}
\toprule
Symbol & Description\\
\midrule
$N$ & Set of data providers, $|N| = n$\\
$D_i$ & Private dataset held by agent $i$\\
$c_i$ ($\hat c_i$) & True (reported) opportunity cost of $D_i$\\
$q_i\!\in\![0,1]$ & Verifiable composite quality score of $D_i$\\
$\widehat{\phi}_i$ & Estimated marginal utility of $D_i$\\
$\psi_i$ & Virtual cost $\frac{\hat c_i}{q_i^\eta(\widehat{\phi}_i+\kappa)^{\gamma}}$\\
$B$ & Planner’s monetary budget (Q-MIA)\\
$U$ & Future utility pool (MUT / Mixed-MIA)\\
$\rho$ & Liquidity factor in Mixed-MIA ($0\le\rho\le1$)\\
$S^\star$ & Selected winner set returned by the mechanism\\
$p_i$ & Final payment to agent $i$\\
$s_i$ & Utility-share weight in token contracts\\
\bottomrule
\end{tabular}
\end{table}

\section{Pseudocode Listings}
\label{app:pseudocode}

\subsection{Quality-Weighted Marginal-Incentive Auction (Q-MIA)}

\begin{algorithm}[H]
\caption{Q-MIA}
\label{alg:qmia}
\begin{algorithmic}[1]
\Require Budget $B$, parameters $(\eta,\gamma,\kappa)$, \\
\hfill\phantom{Require }agents’ reports $\{\hat c_i, q_i, \widehat{\phi}_i\}_{i=1}^n$
\Ensure Winner set $S^\star$, payments $\{p_i\}$
\State Compute $\psi_i\leftarrow\hat c_i\,/\,\bigl(q_i^\eta(\widehat{\phi}_i+\kappa)^\gamma\bigr)$
\State Sort agents by ascending $\psi_i$
\Statex \textit{Greedy selection}
\State $S^\star\leftarrow\emptyset$, $C\leftarrow0$
\For{agent $i$ in sorted order}
    \If{$C+\hat c_i\le B$}
        \State $S^\star\leftarrow S^\star\cup\{i\}$; $C\leftarrow C+\hat c_i$
    \Else
        \State \textbf{break}
    \EndIf
\EndFor
\Statex \textit{Critical payments}
\For{$i\in S^\star$}
    \State Find smallest $\tilde c$ s.t.\ $i$ would still be selected
    \State $p_i\leftarrow \rho\,\tilde c + (1-\rho)\,U\cdot s_i$ \Comment{$\rho{=}1$ in pure Q-MIA}
\EndFor
\end{algorithmic}
\end{algorithm}

\subsection{Mixed-MIA Allocation of Hybrid Rewards}

\begin{algorithm}[H]
\caption{Mixed-MIA (high-level)}
\label{alg:mixed}
\begin{algorithmic}[1]
\Require Budget $B$, utility pool $U$, liquidity $\rho$
\State Run Algorithm~\ref{alg:qmia} with $\rho=1$ to obtain $\{p_i^{Q}\}$
\State Compute $s_i=\dfrac{q_i\widehat{\phi}_i}{\sum_{j\in S^\star}q_j\widehat{\phi}_j}$
\For{$i\in S^\star$}
    \State $p_i\leftarrow\rho\,p_i^{Q}+(1-\rho)\,U\,s_i$
\EndFor
\end{algorithmic}
\end{algorithm}

\section{Quality Metric Implementation Details}
\label{app:quality}

\paragraph{Cleanliness.}  
Count character-level corruption and dictionary mismatches after normalisation; compute $\mathrm{Clean}=1-\min(1,k\!\cdot\!\text{error\_rate})$ with $k=50$.

\paragraph{Diversity.}  
Embed each document with a pretrained MiniLM;
cluster using $k$-means ($k=50$) and evaluate
$H=\!-\,\sum_{j=1}^{k}p_j\log p_j$; normalise to $[0,1]$.

\paragraph{Novelty.}  
Compute cosine similarity $\sigma$ between each sentence and a 5-million-sentence public corpus.  
$\mathrm{Novelty}=1-\mathrm{mean}(\sigma)$; clipped to $[0,1]$.

\paragraph{Metadata Richness.}  
\(
\text{Meta}= \frac{\text{\# filled fields}}{\text{\# required fields}}.
\)

The final score is $q_i=w_1\mathrm{Clean}+w_2 H+w_3\mathrm{Novelty}+w_4\text{Meta}$ with $(w_1,w_2,w_3,w_4)=(0.25,0.25,0.25,0.25)$ unless otherwise stated.

\section{Marginal Utility Estimation Pipeline}
\label{app:marginal}

\begin{enumerate}
\item Train base LLM for $e_0$ epochs on public seed corpus.
\item \textbf{Influence-Function Approximation:}\\
For each $i$, compute $\widehat{\phi}_i\approx -g_i^\top H^{-1}\bar g$ where  
$g_i$ is the gradient of loss w.r.t.\ parameters on $D_i$ and
$H^{-1}\bar g$ is pre-computed with Hutchinson trace estimation.
\item Rescale $\widehat{\phi}_i$ linearly to $[0,1]$ and set floor $10^{-3}$.
\item Log intermediate tensors and random seeds for auditability.
\end{enumerate}

\section{Ethical and Societal Impact Statement}
\label{app:ethics}

Our mechanisms seek to democratise access to LLM training by rewarding diverse data owners.  Nevertheless, selection biases in $q_i$ or $\widehat{\phi}_i$ might disadvantage low-resource languages or under-represented domains.  
To mitigate this risk we (i) publish scoring code, (ii) allow community veto on metric weights, and (iii) reserve 10 \% of budget for a \emph{diversity tranche} prioritising under-served providers.

\vspace{.5em}
\noindent\textbf{Privacy.}  All quality evaluation can be executed under differential privacy or homomorphic hashing; raw data remain encrypted until a binding contract is signed.

\end{appendices}

\end{document}